\documentclass{article}
\PassOptionsToPackage{numbers, compress}{natbib}

\usepackage[preprint]{neurips_2026}
\usepackage{amsmath}
\usepackage{subcaption}
\usepackage{graphicx}
\usepackage{comment}
\usepackage{wrapfig}

\PassOptionsToPackage{numbers,compress}{natbib}
\usepackage{microtype}

\input{commands.sty}

\title{Stochastic Matching via Local Sparsification}

\author{
  Sara Ahmadian \\
  Google Research \\
  \texttt{sahmadian@google.com} \\
  \And
  Edith Cohen \\
  Google Research and Tel Aviv University\\
  \texttt{edith@cohenwang.com} \\
  \And
  Mohammad Roghani \\
  Google Research \\
  \texttt{roghani@google.com} \\
}

\begin{document}

\maketitle
\begin{abstract}
    The classic online stochastic matching problem typically requires immediate and irrevocable matching decisions. However, in many modern decentralized systems such as real-time ride-hailing and distributed cloud computing, the primary bottleneck is often local communication bandwidth rather than the timing of the match itself. We formalize this challenge by introducing a two-stage local sparsification framework. In this setting, arriving requests must prune their realized compatibility sets to a strict budget of $k$ edges before a central coordinator optimizes the global matching. This creates a "middle ground" between local information constraints and global optimization utility.

    We propose a local selection strategy, parametrized by a fractional solution of the expected instance. Theoretically, we quantify the approximation ratio as a function of the solution's {\em spread}. We prove that under sufficient spread, our sparsifier globally preserves the expected size of the maximum matching. Empirically, we demonstrate the robustness of our approach using the New York City ride-hailing datasets and adversarial synthetic benchmarks. Our results show that near-optimal global matching is achievable even with highly constrained local budgets, significantly outperforming standard online baselines. 
\end{abstract}

\section{Introduction}

Online bipartite matching is a fundamental problem in resource allocation and algorithm design. In its standard mathematical framework, the problem is modeled on a bipartite graph $G = (U \cup V, E)$, where a set of offline resources $V$ (e.g., servers, ad inventory, or ride-hailing drivers) is known in advance. Online requests $U$ arrive sequentially over time, revealing their compatible neighborhood $\Gamma(u) \subseteq V$. In the traditional online paradigm, the algorithm must make an immediate and irrevocable decision to either match $u$ to an available resource or drop the request forever.  

The seminal work of Karp, Vazirani, and Vazirani \citep{karp1990optimal} established that randomized ranking achieves a tight $1-1/e$ competitive ratio under strictly adversarial arrivals. To bypass this barrier, a rich line of literature has explored stochastic matching models. In the {\em known distribution} model, request types are drawn independently from a known prior distribution \citep{feldman2009online, manshadi2011online, jaillet2014online, bahmani2010improved}. By leveraging this statistical knowledge, algorithms can successfully break the $1-1/e$ barrier. Subsequent research has generalized these arrival processes, refined approximation bounds, and relaxed the required knowledge of the distributions \citep{bahmani2010improved, karande2011online, gamlath2019online}.

A defining characteristic of this classic literature is the strict requirement of immediate, irrevocable decision-making. However, in many modern decentralized systems, the primary bottleneck is not the irrevocability of a match, but rather {\em local communication bandwidth and memory limits}. In distributed cloud computing, querying the real-time availability of every compatible server incurs prohibitive control-plane latency. Similarly, in ride-hailing or real-time ad exchanges, platform interfaces and latency constraints restrict the system to querying or presenting only a small budget of options per user. In these settings, an algorithm must act locally to narrow down the potential edges before a central coordinator can optimize the global matching.

In this paper, we depart from the strict online matching paradigm to study stochastic matching through the lens of {\em local sparsification}. We formalize this as a two-stage stochastic process:
\begin{enumerate}
\item {\bf Local Pruning:} Each arriving request $u_i \in U$ observes its compatible resources and must immediately apply a local selection rule to choose a subset $S_i$ of exactly $k$ edges, relying only on the prior distribution without coordination with other arrivals.
\item {\bf Global Matching:} Once all $n$ selections are finalized, a central coordinator receives the resulting sparse subgraph $G_S = (U \cup V, \bigcup S_i)$ and computes a maximum bipartite matching $M(G_S)$.
\end{enumerate}
The objective is to design a selection rule that maximizes the {\em preservation ratio} $\alpha = \mathbb{E}[|M(G_S)|] / \mathbb{E}[|M(G)|]$. This framework isolates a fundamental challenge: how to design capacity-constrained local strategies that globally preserve the maximum matching size under uncertainty.

  \subsection{Technical Overview}
 Our approach parameterizes local edge-selection using offline statistical knowledge. We first relax the stochastic problem via the {\em Expected Instance Linear Program} (LP) to compute a fractional matching that maximizes the expected matching size over the known distribution of request types. In the online phase, arriving requests utilize {\em Variance-Optimal} (\varopt) dependent sampling to independently select $k$ edges, using the LP's fractional values as sampling weights. This ensures each request respects the global capacities defined by the fractional plan while making exactly $k$ local choices.

To bound the performance of this sparsifier, we introduce a {\em Heavy-Light} decomposition of the fractional solution. Our core theoretical insight is that the stochastic collision penalty, which limits traditional online algorithms to $1-1/e$ efficiency, is confined exclusively to the "heavy" edges (those with fractional values $> 1/k$). In contrast, "light" mass is preserved by the sparsifier essentially without loss. Consequently, if the fractional matching is "well-spread", we can achieve near-optimal matching with a very small budget. Specifically, for any $\epsilon > 0$, a local budget of $k = \epsilon^{-2}$ guarantees an expected matching size within a $1-\epsilon$ factor of the maximum matching in the realized graph (see \Cref{cor:eps-approx} for the formal statement and proof).
 This effectively bypasses the traditional theoretical limits of strict online matching, motivating our algorithmic objective: to construct highly spread global optimal solutions that guide local choices.
 
This well-spread assumption naturally captures the reality of dense platforms such as ride-hailing, where riders are surrounded by many nearly interchangeable vehicles rather than relying on a single critical connection (see \Cref{claim:equivalence} for the formal definition of interchangeable resources). Guided by the theory, we also propose a simple Monte Carlo heuristic for constructing high-spread fractional solutions: we repeatedly sample
realizations of the stochastic instance, compute offline optimal matchings under random shuffling, and average the resulting edge
incidences to obtain fractional selection weights. These weights are then used by the local VarOpt sparsifier in our experiments.

We validate this using NYC Taxi data, showing that our local sparsifier closely tracks the performance of a centralized, full-information offline maximum matching. While this well-spread property is common in practice, our framework remains robust even when this structural assumption is violated. In several of our adversarial synthetic benchmarks, heavy edges make up a significant portion of the fractional solution, yet the algorithm continues to perform well even without this assumption. Crucially, our results on these synthetic benchmarks demonstrate that by returning a subgraph of size $k > 1$, our algorithm fundamentally bypasses the theoretical $0.901$ efficiency barrier established for traditional online stochastic matching algorithms. For more information about the experiments, see \Cref{sec:experiments}.

\subsection{Related work}
\paragraph{Online Random-Order Arrival.} To bypass the $1-1/e$ adversarial barrier, a line of literature has studied random order arrivals. Under this condition, the RANKING algorithm achieves an approximation ratio between 0.696 and 0.727 for unweighted bipartite graphs \citep{manshadi2011online, karande2011online}. For vertex-weighted variants, RANKING achieves 0.686 \citep{huang2019online, jin2021improved, peng2025revisiting}, while edge-weighted cases reach 0.659 \citep{HuangSWZ25}. For general graphs, a series of works has steadily advanced our understanding, recently improving the approximation ratio for the RANKING algorithm under random vertex arrivals to 0.56 \citep{AronsonDFS95, chan2014ranking, 0.546ranking, poloczek2012randomized, angW020, 0.505_simplified_Ranking, derakhshan2026unified}. Weighted variants of randomized greedy matching in general graphs remain particularly challenging; algorithms have recently been shown to achieve ratios of 0.5015 for the vertex-weighted case \citep{chan2018analyzing} and 0.5014 for the edge-weighted variant \citep{angW020}

\paragraph{Streaming and Communication Complexity.} A related line of research examines maximum matching in the streaming setting under vertex arrivals \citep{kapralov2014approximating, goel2012communication, kapralov2013better}. While streaming algorithms can store $O(n \text{polylog}(n))$ edges, no single-pass algorithm under adversarial arrivals can exceed the $1 - 1/e$ barrier \citep{kapralov2013better}. For edge-arrival streams, research typically relies on constructing sparsifiers to preserve matching size \citep{kapralov2021space, bernstein2024improved, assadi2021beating}. Our work differs by requiring that the sparsification happens {\em locally} and {\em independently} at the time of vertex arrival, rather than having access to the full stream history or global coordination.

\paragraph{The Power of Spread and Regularity.} The utility of spread solutions has been implicitly recognized across several matching paradigms. In online matching, the Two Suggested Matchings (TSM) algorithm uses $k=2$ choices for routing preferences \citep{feldman2009online, karande2011online}, while $k$-regular fractional solutions allow online algorithms to achieve $1-O(1/\sqrt{k})$ efficiency \citep{bahmani2010improved}. Similar structural guarantees are used to achieve sublinear time in offline matching for $d$-regular graphs \citep{goel2010perfect, dani2025sublinear}. The importance of this structure is underscored by recent hardness results: without such guarantees, even estimating the size of a $(1-\epsilon)$-approximate maximum matching requires near-quadratic time \citep{behnezhad2024approximating, earthmoverdistance}. Our framework formalizes the intuition that when weights are sufficiently dispersed, localized sampling provides a powerful mechanism to recover near-optimal global matchings without the need for exhaustive exploration.

\section{Problem Definition and the Global Baseline}
\label{sec:problem_definition}
We formalize the challenge of matching under communication constraints as a two-stage stochastic process. Our goal is to design a local decision rule that preserves the global structural properties of a random graph while adhering to a strict information budget.

\subsection{The Stochastic Matching environment}

We operate under the \emph{known distribution} model, a standard benchmark in stochastic optimization \citep{feldman2009online, manshadi2011online}. An instance is defined by $(V, \mathcal{D}, n)$, where:
\begin{itemize}
    \item \textbf{Supply ($V$):} A fixed set of resources $V = \{v_i\}$. 
    \item \textbf{Demand Distribution ($\mathcal{D}$):} A known probability distribution over $2^V$ possible subsets of $V$ representing different demand types. 
    \item \textbf{Number of requests ($n$):} The total number of requests that will arrive.
\end{itemize}

The demand is realized as a set $U = \{u_1, \dots, u_n\}$ where each request $u_i$ is independently associated with a random subset of compatible resources $R_i \subseteq V$ sampled {\em i.i.d.} from the distribution $\mathcal{D}$. This process induces a {\em realized bipartite graph} $G = (U \cup V, E)$, where an edge $(u_i, v)$ exists if and only if $v \in R_i$.

\subsection{Local Sparsification} 

In many real-world systems, the full graph $G$ is too large to process centrally due to latency or bandwidth limits. We introduce a {\em local budget constraint} $k$, which dictates that each request $u_i$ can only "report" a subset of its compatible edges to the central coordinator.

\begin{definition}[The Selection Rule $\Phi_k$]
A sparsifier is defined by a local randomized selection function $\Phi_k: 2^V \to \Delta(2^V)$, where $\Delta(2^V)$ denotes the space of probability distributions over the subsets of $V$. For each arriving request $u_i$ with compatibility set $R_i$, 
the function draws a random subset $S_i \sim \Phi_k(R_i)$ such that $S_i \subseteq R_i$ and $|S_i| \leq k$.

Crucially, $\Phi_k$ must be \textbf{local}: the selection for request $u_i$ is made using only the distribution $\mathcal{D}$ and the realization $R_i$, without any knowledge of other requests $u_j$ ($j \neq i$).
\end{definition}

\subsection{Maximizing the Preservation Ratio}
The budget $k$ forces a form of strategic pruning. If $k=1$, the request must commit to a single resource immediately, reverting to the classic online matching paradigm where a $1-1/e$ competitive ratio is the absolute barrier. As $k$ increases, we move into a "middle ground" where we aim to maintain the global utility of the graph with minimal local information. 

Let $G_S = (U \cup V, E_S)$ be the sparsified graph with edge set $E_S = \bigcup_{i=1}^n \{(u_i, v) : v \in S_i\}$. We measure the performance of a selection strategy by the {\em preservation ratio} $\alpha$:
\begin{equation}
    \alpha = \frac{\mathbb{E}[|M(G_S)|]}{\mathbb{E}[|M(G)|]}
\end{equation}

where $M(G)$ denotes the size of the maximum matching in graph $G$. Our goal is to design a rule $\Phi_k$ that maximizes $\alpha$.

\subsection{The Expected Instance LP} 
Our framework relies on the {\em Expected Instance Linear Program} (LP), a well-established formulation that captures the average-case behavior of the stochastic instance by fractionally routing expected demand. This LP serves as a global roadmap, enforcing matching constraints while relaxing the discrete arrival process.

For an instance of $(V, \mathcal{D},n)$,
\begin{itemize}
    \item Let $T = \{t_1, \dots, t_m\}$ be the set of all possible request types.
    \item Let  $p_j$ be the single-draw probability of type $t_j$, where $\sum_{j=1}^m p_j = 1$. The expected number of arrivals for type $t_j$ is $n p_j$.
    \item Let $\Gamma(t_j) \subseteq V$ be the subset of resources compatible with type $t_j$.
\end{itemize} 

\begin{definition}[Expected Instance LP] \label{def:fractionalmatchingLP}
Let $x_{ij}$ be a fractional variable representing the conditional probability that a request of type $t_j$ is matched to resource $v_i$, therefore the expected match on the edge $(t_j, v_i)$ is $np_jx_{ij}$. The goal is to maximize the total expected matched mass:
\begin{align*}
\text{Maximize} \quad Z(x) &:= \sum_{j=1}^m n p_j \sum_{v_i \in \Gamma(t_j)} x_{ij} \\
\text{subject to:} \quad \quad
\sum_{j: v_i \in \Gamma(t_j)} n p_j x_{ij} &\le 1, \quad \forall v_i \in V \quad &&\text{(Resource Capacity)} \\
\sum_{v_i \in \Gamma(t_j)} x_{ij} &\le 1, \quad \forall t_j \in T \quad &&\text{(Type Matching Limit)} \\
x_{ij} &\ge 0, \quad \forall i, j \quad &&\text{(Non-negativity)}
\end{align*}
\end{definition}

Let $\mathrm{OPT}_{\mathrm{LP}}$ denote the optimal value of the Expected Instance LP. This value is a natural upper bound on the expected size of any matching in the realized graph. Any gap between $\mathrm{OPT}_{\mathrm{LP}}$ and $\mathbb{E}[|M(G)|]$ arises from {\em stochastic collisions}, events where realized demand assigned to a resource (with respect to $x_{ij}^*$) exceeds its capacity. Previous work establishes a tight relationship between these two quantities; for completeness, we provide a proof of the following bound in Appendix \ref{app:lp-bound-expected-matching-size}.  

\begin{prop}[Relation between $\mathrm{OPT}_{\mathrm{LP}}$ and $\E_G{[|M(G)|]}$] \label{prop:lp-bound}
The expected size of the offline maximum matching is bounded by:
\[ \left(1 - \frac{1}{e}\right) \mathrm{OPT}_{\mathrm{LP}} \leq \E_G{[|M(G)|]} \leq \mathrm{OPT}_{\mathrm{LP}} \]

\end{prop}

\section{Algorithmic Framework: From Global Plans to Local Choices}
\label{sec:sparsifier_selection}
To navigate the communication bottleneck, we propose a two-stage strategy: an offline phase that computes a globally robust fractional plan via the Expected Instance LP, and an online phase that uses this plan to guide local edge selection. Our algorithm translates the global fractional plan $x^*$ into local edge selections using Variance-Optimal (\varopt) sampling.

\subsection{Sampling with Hard Constraints: \varopt}
While the LP provides the optimal fractional flow $x^*$,  the local communication budget $k$ requires a discrete selection of edges. A naive approach would be independent Bernoulli sampling, but this lacks control over the exact sample size. We instead utilize Variance-Optimal (\varopt) sampling \citep{varopt2011}, a family of unequal probability dependent sampling schemes \citep{chao1982unequal}.

Given a set of items $e\in E$ with weights $x_e\ge 0$, \varopt~ determines a uniform threshold multiplier $\tau > 0$ and assigns an inclusion probability $\pi_e$ to each item such that:
\begin{equation} \label{eq:varopt_prob}
\pi_e = \min\{1, \tau \cdot x_e\} \quad \text{where} \quad \sum_{e \in E} \pi_e = \min\{k,|E|\}.
\end{equation}
This thresholding mechanism ensures that ``heavy'' items (where $\tau \cdot x_e \geq 1$) are included deterministically ($\pi_e = 1$), while ``light'' items are sampled strictly proportional to their weights. For any subset, the probability of inclusion of an item is bounded below by its proportional share of the budget: $\pi_e \geq \min\left\{1, k \frac{x_e}{\sum_{e' \in E} x_{e'}}\right\}$.

Importantly for applications, \varopt~ sampling is very efficient and requires a single pass over the edges \citep{varopt2011} (and does not require a root-finding search for  $\tau$). 

To provide unbiased estimation of the matching size, each sampled item is assigned an inverse probability weight $w_e$. Let $I_e \in \{0, 1\}$ be the indicator that item $e$ is included in the sample $S$. The \varopt~ weight is defined as:
\[ w_e = \begin{cases} \frac{x_e}{\pi_e} & \text{if } I_e = 1 \\ 0 & \text{otherwise} \end{cases} \]
By design, $\E[w_e] = x_e$. A defining property of \varopt~ (and the reason it is superior to independent Bernoulli sampling for our purposes), is that the sum of these weights is deterministically preserved. For any realization of the sample:
\begin{equation} \label{eq:varopt_sum}
    \sum_{e \in S} w_e = \sum_{e \in S} \frac{x_e}{\pi_e} = \sum_{e \in E} x_e.
\end{equation}

This property removes the variance associated with the random size of the sample, providing tighter empirical concentration while strictly respecting the hard capacity limit $k$.

In our subsequent analysis, we leverage two fundamental structural properties of \varopt:
\begin{enumerate}
    \item \emph{Negative Association and Concentration:}  The indicator variables $I_e$ are negatively associated. Consequently, the variance of the total load is sub-additive, and standard concentration inequalities (e.g., Chernoff-Hoeffding) hold at least as tightly as they would under independent sampling \citep{Shao:2000,panconesi1997randomized, srinivasan2001distributions}.
    \item \emph{Composition of Samples:} \varopt~ sampling schemes cleanly compose. If a dataset is partitioned and we apply a \varopt~ sample over the partitions, followed by a local \varopt~ sample within the selected partitions, the resulting global distribution over the individual items preserves the exact sample size constraint and the negative association properties of a single, unified \varopt~ process.
\end{enumerate}

\subsection{Sparsifier Selection and the Role of Solution Geometry}
Given a feasible solution $x$ to the Expected Instance LP, we define our local selection strategy. This strategy effectively "rounds" the fractional roadmap provided by the LP into a sparse realized graph.

\begin{definition}[Sparsifier Selection Strategy $\Phi_{k,x}$] \label{def:selection-strategy}
When a request $u_r$ of type $t_j$ arrives, the algorithm observes its realized compatibility set $\Gamma(t_j)$. The strategy $\Phi_{k,x}$ selects a subset of edges $S_r \subseteq \{(u_r,v)\mid v\in \Gamma(t_j)\}$ by applying  $\varopt$ sampling to the fractional values $x_{ij}$ as the sampling weights.
\end{definition}

The performance of this local rule is intrinsically tied to the "shape" of the LP solution. While any optimal solution $x$ yields the same global objective $Z(x) = \mathrm{OPT}_{\mathrm{LP}}$, the local budget $k$ introduces a sensitivity to how fractional mass is distributed. Because requests cannot coordinate to avoid resource collisions, the strategy benefits from solutions that are "spread out" rather than "concentrated."

\begin{observation}[The Importance of the Fractional Geometry] \label{obs:solution-geometry}
Consider a fully connected bipartite instance with $n$ resources and $n$ uniform types ($p_j = 1/n$). 
\begin{itemize}
    \item \textbf{A Concentrated Solution:} If the LP sets $x_{ij} = 1$ for $i=j$ and $0$ otherwise, each type is tied to a specific resource. When a type $t_j$ arrives, the local sparsifier $\Phi_{k,x}$ is forced to pick this single edge regardless of the budget $k$, and the expected size of the maximum matching is exactly the number of distinct types that arrive, which is $n \Pr(\texttt{an item appearing}) = n (1-(1-1/n)^n) \approx n(1-1/e)$. 
    \item \textbf{A Spread Solution:} If the LP sets $x_{ij} = 1/n$ for all $i,j$, the fractional weights are small. The sparsifier $\Phi_{k,x}$ samples a random set of $k$ edges per arrival. As $k$ increases, the probability that a single resource is ignored or overwhelmed decreases, therefore the expected size of the matching approaches $n$ and the preservation ratio approaches $1$.
\end{itemize}
\end{observation}

This observation motivates the need to carefully select the feasible solution $x$ so that ``light'' fractional values that spread demand and mitigate stochastic collisions are preferred, leading to the following decomposition. 

\subsection{Quantifying Sparsifier Quality: Heavy-Light Decomposition}
To bridge the gap between the local budget $k$ and global performance, we partition the edges in the support of $x$ based on their ``value'' relative to the budget's inverse, $1/k$.
\begin{definition}[Heavy-Light Decomposition] \label{def:light-heavy}
Given a sparsification parameter $k$ and a feasible fractional solution $x$, we partition the edges into: 
\begin{itemize}
    \item \emph{Light Edges ($\mathcal{E}_L$):} $\{(i, j) \mid x_{ij} \leq 1/k\}$. These represent demand that can be spread across the budget.
    \item \emph{Heavy Edges ($\mathcal{E}_H$):} $\{(i, j) \mid x_{ij} > 1/k\}$. These represent concentrated demand that the budget $k$ cannot further dilute.
\end{itemize}
This partition induces a natural decomposition of the expected fractional objective into $Z(x) = Z_L(x) + Z_H(x)$, where:
\[ Z_L(x) = \sum_{(i,j) \in \mathcal{E}_L} n p_j x_{ij} \quad \text{and} \quad Z_H(x) = \sum_{(i,j) \in \mathcal{E}_H} n p_j x_{ij} \]
\end{definition}

We now establish that our sparsifier preserves the light fractional components, confining the stochastic collision penalty to the heavy components. Here, we present our main result and defer the proof to Appendix \Cref{app:theorem-proof}.

\begin{theorem}[Sparsifier Approximation Guarantee] \label{thm:approx-ratio}
Let $G_S$ be the sparsifier constructed by the local $\Phi_{k,x}$ selection strategy from a feasible solution $x$, using a budget parameter $k$. The expected size of the maximum matching on $G_S$ satisfies:
\[
\small
\E[|M(G_S)|] \geq Z\left(1-  \min\left\{
\frac12\sqrt{\frac{Z_H}{Z}+\frac{1}{k}\frac{Z_L}{Z}},
\;
e^{-\frac{Z_H}{Z}}
\left(
\frac{Z_H}{Z}+\frac12\sqrt{\frac{1}{k}\frac{Z_L}{Z}}
\right)
\right\}  \right)-1
\]
\end{theorem}

\begin{corollary}[Near-Optimal Sparsification] \label{cor:eps-approx}
If the fractional solution $x$ is sufficiently spread such that the heavy component constitutes at most an $\epsilon/2$ fraction of the total objective (i.e., $Z_H(x)/Z(x) \leq \epsilon/2$), then a sparsifier budget of $k = \epsilon^{-2}$ guarantees that \[\frac{\E[|M(G_S)|]}{Z(x)} \geq 1 - \epsilon - \frac{1}{Z(x)}\ .\] In particular, if $Z(x)\geq \E[|M(G)|]$, the preservation ratio is
$\frac{\E[|M(G_S)|]}{\E[|M(G)|]} \geq 1 - \epsilon - \frac{1}{\E[|M(G)|]}$.
\end{corollary}
\begin{proof}
Substituting $Z_H(x)/Z(x) \leq \epsilon/2$ and $1/k = \epsilon^2$ into the bound from Theorem \ref{thm:approx-ratio}.
\end{proof}

In practice, it is often the case that large sets of resources are highly symmetric or interchangeable. The following claim demonstrates that we can exploit these symmetries to naturally spread the fractional demand, shifting weight from $Z_H$ to $Z_L$.

\begin{claim}[Spread Solutions via Equivalence Classes] \label{claim:equivalence}
We say two resources $v, v'$ are interchangeable if for all types $t_j \in T$, they share the exact same compatibility: $v \in \Gamma(t_j) \iff v' \in \Gamma(t_j)$. If an instance contains an equivalence class of interchangeable resources of size $\ell$, there exists an optimal fractional solution $x^*$ where all edges incident to this class are $\ell$-light (i.e., $x^*_{ij} \leq 1/\ell$).
\end{claim}
\begin{proof}
Let $x$ be any optimal solution. For any type $t_j$ and the equivalence class $C$ of size $\ell$, we construct $x^*$ by uniformly averaging the fractional weights across the class: $x^*_{ij} = \frac{1}{\ell} \sum_{v \in C} x_{v, j}$ for all $v_i \in C$. Because the resources are strictly interchangeable, this uniform redistribution preserves both the capacity constraints and the global objective value $Z(x)$. Furthermore, since the type matching limit enforces $\sum_{v \in C} x_{v, j} \leq 1$, the averaged weight on any single edge incident to the class is bounded by $x^*_{ij} \leq 1/\ell$.
\end{proof}

Motivated by Theorem \ref{thm:approx-ratio}, our algorithmic objective is to construct solutions $x$ with a favorable spread, and in particular, aim to maximize the light component $Z_L(x)$.

\begin{figure}[t!]
    \centering

    \begin{subfigure}[b]{0.98\columnwidth}
        \centering
        \includegraphics[width=0.8\linewidth]{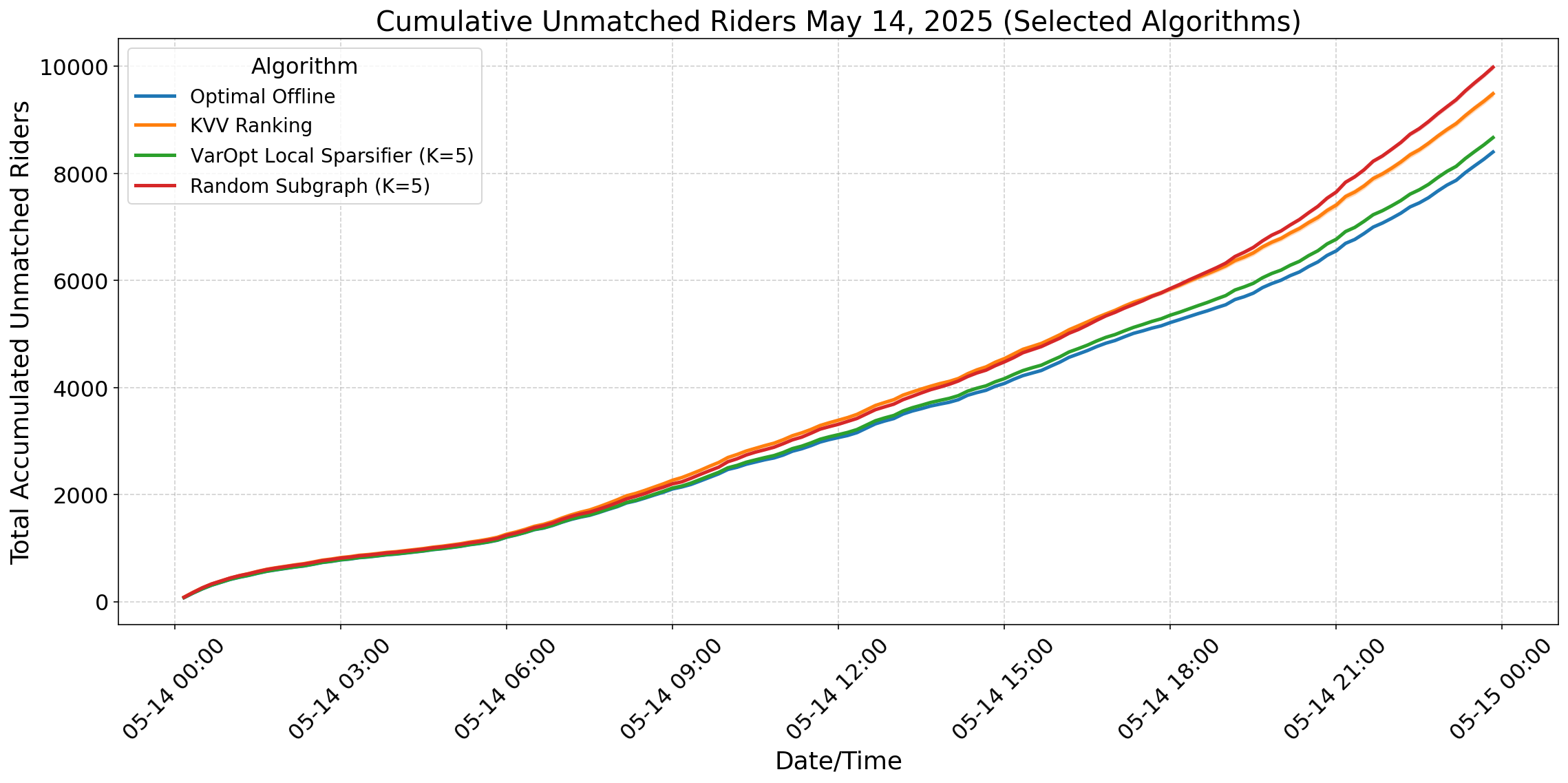}
        \caption{Baseline comparison of cumulative unmet demand. The \varopt~local sparsifier ($k=5$) outperforms uniform random selection and KVV baseline and tracks closely with the full-information optimal offline.}
        \label{fig:main_comp}
    \end{subfigure}

    \vspace{0.25cm}

    \begin{subfigure}[b]{0.98\columnwidth}
        \centering
        \includegraphics[width=0.8\linewidth]{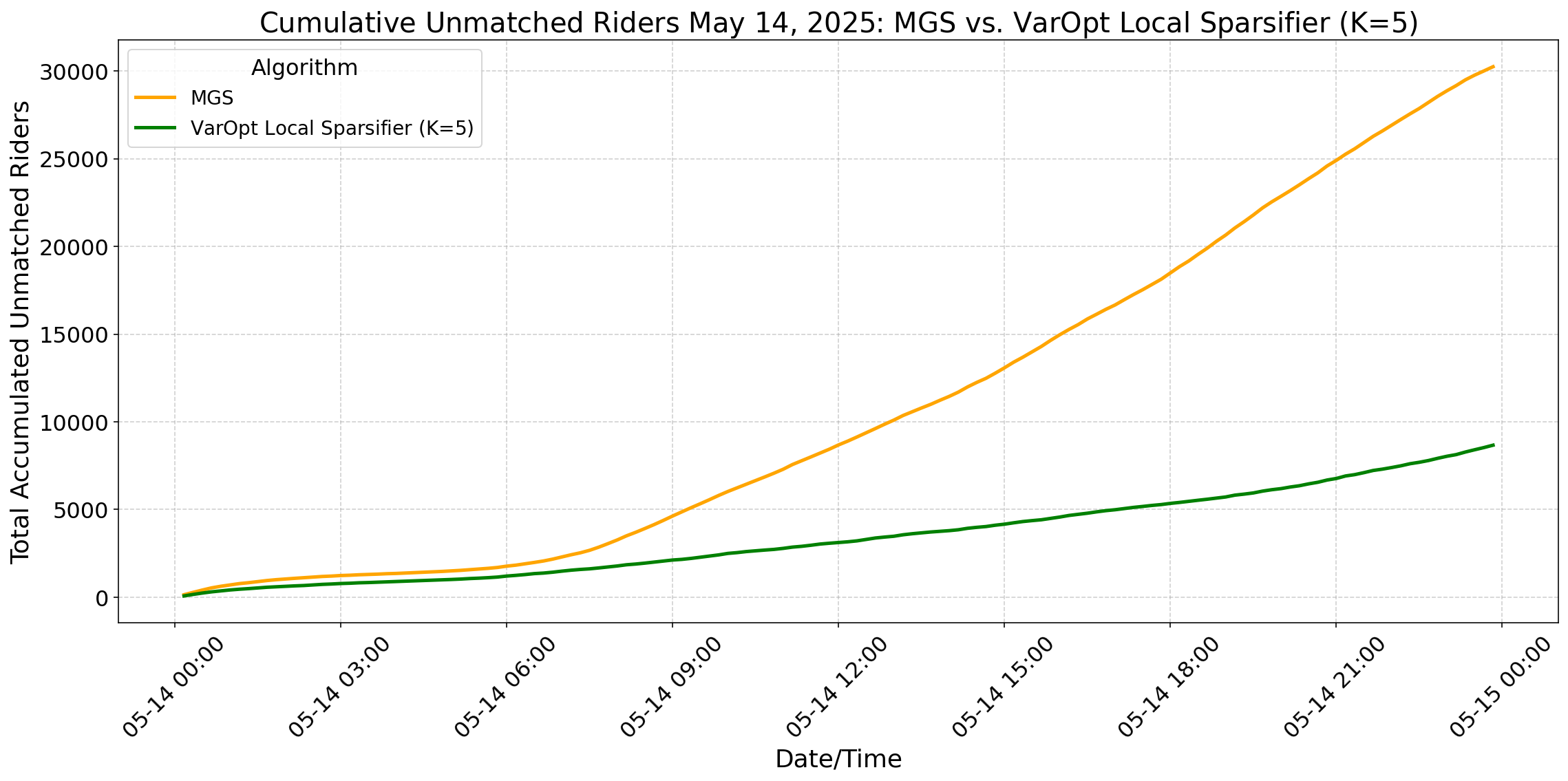}
        \caption{Performance against the standard MGS algorithm.}
        \label{fig:mgs_comp}
    \end{subfigure}

    \vspace{0.25cm}

    \begin{subfigure}[b]{0.98\columnwidth}
        \centering
        \includegraphics[width=0.8\linewidth]{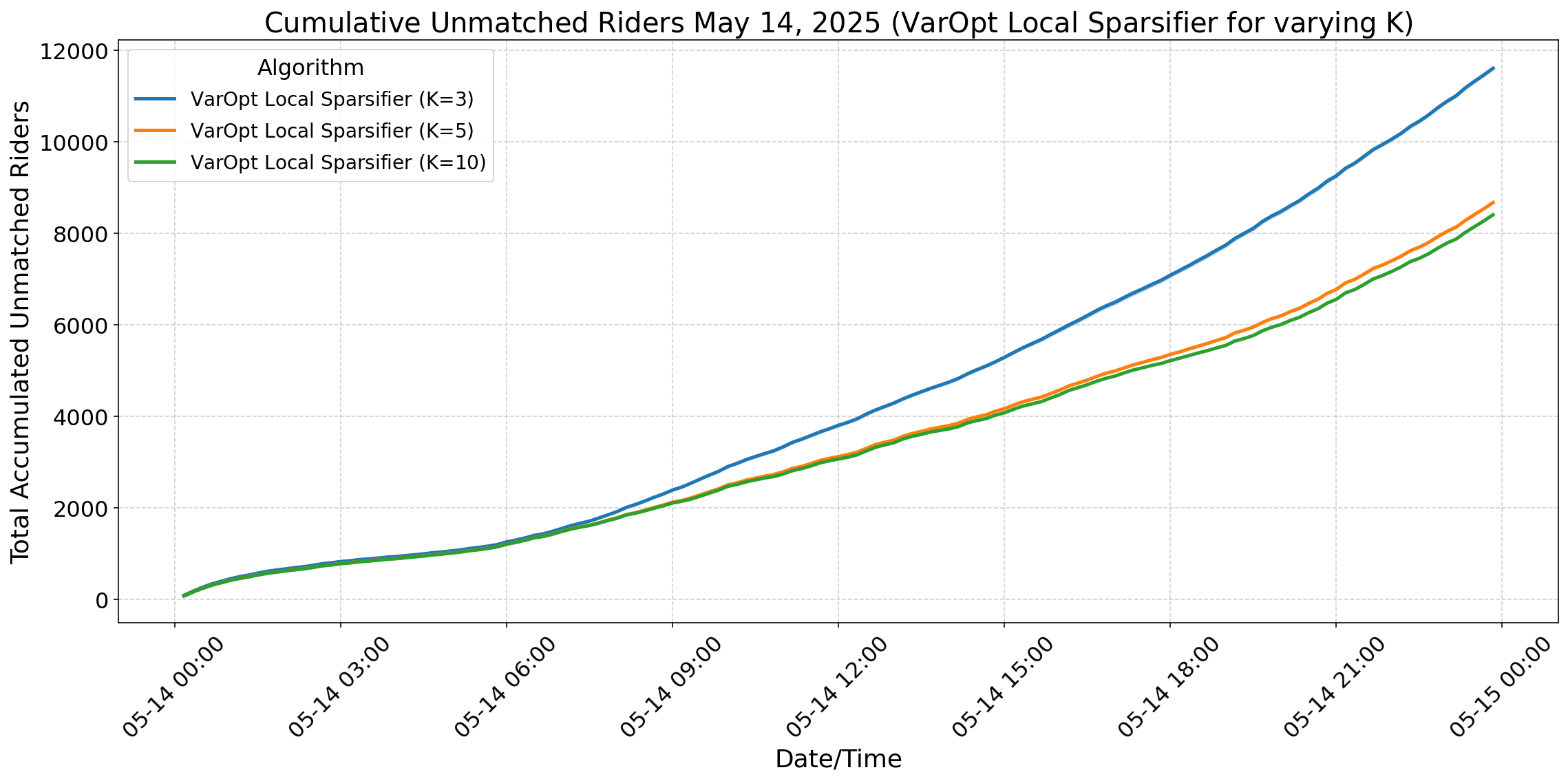}
        \caption{Ablation study on the local sparsifier size ($k$) for different values of $k \in \{3,5,10\}$. Increasing the budget from $k=3$ to $k=10$ yields marginal gains, confirming that a small subset of highly probable edges is sufficient.}
        \label{fig:k_ablation}
    \end{subfigure}

    \caption{Empirical evaluation of stochastic matching algorithms on NYC taxi data (May 14, 2025). The plots are reported with $\pm$95\% Confidence Intervals. Because the primary metric aggregates cumulative unmet demand over time, the variance remains exceedingly small relative to the total volume. Consequently, these confidence bands are tight and visibly indistinguishable from the mean lines.}
    \label{fig:nyc_results}
\end{figure}

\section{Experiments}\label{sec:experiments}

In this section, we evaluate the performance of our stochastic matching sparsifiers across two distinct environments: (1) a large-scale ride-hailing simulation using real-world taxi data, and (2) a suite of synthetic adversarial benchmarks designed to stress-test online matching algorithms.  Our objective is to demonstrate that by restricting the available edges for each rider to a small set $k$, we can recover a significant fraction of the offline optimal matching. We compare our approach against standard online benchmarks and heuristic sparsification strategies to evaluate how guided edge-selection impacts market efficiency. 

\subsection{Real-World Application: NYC Ride-Hailing}

\textbf{Dataset and Graph Construction.} We utilize the NYC Yellow Taxi Trip Data from May 14, 2025\footnote{https://www.nyc.gov/site/tlc/about/tlc-trip-record-data.page}. The spatial framework is discretized into taxi zones, with an adjacency matrix $A$ where $A_{uv} = 1$ if zones $u$ and $v$ share a boundary.  

The simulation proceeds in 10-minute intervals. For each interval $t$, we construct a bipartite graph $G_t = (U_t \cup V_t, E_t)$ where:

\begin{itemize}

    \item $V_t$ (Offline/Supply): Vehicles sampled from drop-off events in $[t-5\text{m}, t]$.

    \item $U_t$ (Online/Demand): Riders sampled from pick-up events in $[t, t+5\text{m}]$.

    \item $E_t$: An edge exists if the car's drop-off zone and the rider's pick-up zone are identical or adjacent.
\end{itemize}

The construction of graph $G_t$ is detailed in \Cref{sec:graph_construction}.

\textbf{Evaluated Algorithms.} We compare five strategies to observe the impact of limited information and restricted choice:
\begin{enumerate}
    \item \textbf{Optimal Offline:} The maximum matching $M(G_t)$ computed via Hopcroft-Karp with full knowledge of the interval. This is the absolute upper bound.
    \item \textbf{KVV Ranking \cite{karp1990optimal}:} The classic online algorithm assigning random ranks to offline cars; arriving riders greedily pick the highest-ranked available neighbor.
    \item \textbf{Random Subgraph:} A naive sparsifier where riders only consider a random subset of at most $k$ valid edges.
    \item \textbf{MGS Algorithm \cite{manshadi2011online}:} Following Algorithm 1 in \cite{manshadi2011online}, this approach uses offline statistics to guide online matching decisions. Specifically, it computes the marginal probability that each edge appears in an optimal offline matching and then it samples two matchings as guidance for the online algorithm.
\item \textbf{\varopt\ Local Sparsifier:} Our proposed method.  We estimate the conditional match probabilities $x_{ij} = \mathbb{P}((u_j, v_i) \in M(G_t))$ by averaging the results of 100 Monte Carlo simulations (details provided in \Cref{sec:MonteCarlo}). Arriving riders then apply \varopt~sampling (see  \Cref{alg:online_sparsifier}) to select $k$ edges based on these weights.  
\end{enumerate}

\textbf{Results.} Our primary metric is the volume of unmatched riders (unmet demand). As shown in Figure \ref{fig:nyc_results}, the \textit{\varopt~ Local Sparsifier} outperforms the other baseline algorithms in this metric. Furthermore, when $k=10$, its performance closely approaches the offline optimal.

Figure \ref{fig:nyc_results} presents a comprehensive breakdown of the algorithms' performance across the 24-hour simulation period. In Figure \ref{fig:main_comp}, we observe that restricting choices via the \textit{\varopt~Local Sparsifier} ($k=5$) yields significantly fewer unmatched riders than the naive \textit{Random Subgraph} ($k=5$) and KVV algorithm. Figure \ref{fig:mgs_comp} highlights the inefficiency of the MGS baseline in this high-variance setting, where our proposed sparsifier demonstrates vastly superior matching rates. To evaluate the sensitivity of our approach to the choice $k$, Figure \ref{fig:k_ablation} compares varying sizes of the local sparsifier. The gap between $k=5$ and $k=10$ is minimal, reinforcing our theoretical claim that a small $k$ is sufficient to capture the majority of the optimal matching weight.

\subsection{Synthetic Tests}
We evaluate the algorithms against structured hard examples from the literature of online matching. These adversarial benchmarks are designed to highlight the failure modes of standard algorithms and demonstrate the robustness of probabilistic edge selection.

\textbf{The Partitioned Block Graph (Karande et al. \cite{karande2011online}):} Following Section 4.2 of \cite{karande2011online}, we construct a bipartite graph over $n$ vertices partitioned into three row blocks ($R_1, R_2, R_3$) and three column blocks ($C_1, C_2, C_3$). The column and row blocks contain $0.3n$, $0.4n$, and $0.3n$ vertices, respectively. Edges exist on the main diagonal ($i=j$), between $R_1$ and $C_2$, and between $R_2$ and $C_3$. Arriving columns are random, while rows are randomly permuted. This specific block-diagonal structure restricts the KVV Ranking algorithm to a competitive ratio of approximately 0.727 when both sides of the bipartite graph are randomly permuted; however, we emphasize that this tight bound applies specifically to the random order arrival model and does not extend to the i.i.d. arrival setting considered in our work.

\textbf{The TSM Tightness Graph (Feldman et al. \cite{feldman2009online}):} We implement the construction from Section 4.2.5 of \cite{feldman2009online}, originally designed as a tight adversarial instance for the Two-Suggested-Matching (TSM) algorithm. For a graph of size $n$ (a multiple of 4), the structure consists of $n/4$ disjoint 6-cycles ($\{u_i - x_i - v_i - y_i - w_i - z_i - u_i\}$) interspersed with two dense complete bipartite subgraphs (between $K$ and $X$, and $L$ and $W$). The optimal matching achieves an assignment size of $n(1 - 1/2e)$. This structure acts as a trap for offline-trained algorithms, challenging our sparsifier to hold onto the optimal matchings it found during the Monte Carlo phase.

\textbf{The KVV Upper Triangular Graph:} We test the classic adversarial construction for the Ranking algorithm. In this $n \times n$ graph, arriving node $i$ is connected to offline nodes $\{i, i+1, \dots, n\}$. Our \varopt~local sparsifier leverages its offline Monte Carlo weights to heavily prioritize the exact diagonal.

\textbf{The Online Upper-Bound Graph (Bahmani and Kapralov \cite{bahmani2010improved}):} We implement the graph construction from Theorem 7 in \cite{bahmani2010improved}, which establishes a strict theoretical upper bound of $\approx 0.901$ for any traditional online stochastic matching algorithm. The bipartite graph features offline nodes $A_1 \cup A_2$ and online types $I_1 \cup I_2$, sized such that $|I_1| = |A_2| = n$ and $|A_1| = |I_2| = n/e$. The edge set consists of a hidden perfect matching between $I_1$ and $A_2$, obscured by complete bipartite subgraphs between $I_2 \cup A_2$ and $I_1 \cup A_1$. Because traditional online algorithms must irreversibly commit to a single edge upon arrival, they mathematically cannot exceed the 0.901 efficiency limit on this instance. We use this test to explicitly demonstrate that returning a subgraph of $k > 1$ edges fundamentally bypasses the theoretical constraints of strict online commitment.

\begin{table*}[htbp]
    \centering
    \caption{Average matching efficiency (approximation ratio) across synthetic adversarial benchmarks. Each graph consists of 100 vertices per side. Results are averaged over 100 independent trials, with 100 Monte Carlo simulations used for offline learning. Values are reported as Mean $\pm$ 95\% Confidence Interval.}
    \label{tab:synthetic_results}
    \resizebox{\textwidth}{!}{
    \begin{tabular}{l c c c c c c c c}
        \toprule
        & & & \multicolumn{3}{c}{Random Subgraph} & \multicolumn{3}{c}{\varopt~Local Sparsifier} \\
        \cmidrule(lr){4-6} \cmidrule(lr){7-9}
        Adversarial Graph & KVV & MGS & K=3 & K=5 & K=10 & K=3 & K=5 & K=10 \\
        \midrule
        Partitioned Block \cite{karande2011online} & $82.82\% \pm 0.54\%$ & $70.07\% \pm 1.07\%$ & $74.16\% \pm 0.65\%$ & $77.48\% \pm 0.73\%$ & $83.47\% \pm 0.55\%$ & $93.61\% \pm 0.50\%$ & $98.44\% \pm 0.27\%$ & $99.97\% \pm 0.04\%$ \\
        KVV Upper Triangular & $91.28\% \pm 0.45\%$ & $68.53\% \pm 0.89\%$ & $76.27\% \pm 0.53\%$ & $84.53\% \pm 0.56\%$ & $93.33\% \pm 0.42\%$ & $95.21\% \pm 0.40\%$ & $99.11\% \pm 0.16\%$ & $99.98\% \pm 0.03\%$ \\
        Bahmani 0.901 Bound \cite{bahmani2010improved} & $83.40\% \pm 0.53\%$ & $68.78\% \pm 0.72\%$ & $61.77\% \pm 0.77\%$ & $66.49\% \pm 0.87\%$ & $76.80\% \pm 0.84\%$ & $89.69\% \pm 0.44\%$ & $95.11\% \pm 0.40\%$ & $99.26\% \pm 0.17\%$ \\
        TSM Tight \cite{feldman2009online} & $94.58\% \pm 0.42\%$ & $74.93\% \pm 0.95\%$ & $98.13\% \pm 0.44\%$ & $99.61\% \pm 0.20\%$ & $99.93\% \pm 0.08\%$ & $97.54\% \pm 0.50\%$ & $99.58\% \pm 0.19\%$ & $99.96\% \pm 0.05\%$ \\
        \bottomrule
    \end{tabular}
    }
\end{table*}

\paragraph{Computational Resources} All experiments were conducted on a standard cloud-based computing platform. The runtime environment for each simulation was configured with 2 vCPUs and 32 GB of system RAM.

\section{Conclusion and Open Questions}

In this work, we introduced a two-stage local sparsification framework for stochastic bipartite matching, motivated by the strict communication and latency bottlenecks inherent in modern decentralized platforms. We proposed a capacity-constrained selection strategy that leverages dependent sampling, parameterized by the fractional solutions of the expected instance. Theoretically, we formalized a Heavy-Light decomposition to prove that structurally spread fractional demand allows a local budget of $k$ choices to successfully bypass the classic stochastic collision penalty. Empirically, evaluations on both real-world ride-hailing data and adversarial synthetic benchmarks demonstrate that our local sparsifier vastly outperforms random selection baselines and strict online algorithms. 

Our framework opens several compelling directions for future research, primarily concerning the fundamental information-theoretic gap between the maximum matching on the true stochastic realization $\E[|M(G)|]$ and the matching on the sparsifier $\E[|M(G_S)|]$. While our analysis establishes a strong $1-\epsilon$ approximation when the underlying fractional solution is predominantly light, the gap's behavior on graphs with unavoidable heavy edges remains an open challenge. The most intriguing open question is whether there exists a local sparsification strategy, perhaps by applying our framework with a carefully constructed solution, where the ratio $\E[|M(G_S)|] / \E[|M(G)|]$ asymptotically approaches $1$ for larger $k$ across \emph{all} graph topologies.

\newpage
\bibliographystyle{plainnat}
\bibliography{main}

\appendix

\section{Proof of \Cref{prop:lp-bound}}\label{app:lp-bound-expected-matching-size}

The proof is divided into two parts: showing the upper bound via the LP relaxation, and proving the lower bound via an independent randomized assignment. Let $x^*$ be an optimal fractional solution to the Expected Instance LP, such that $\sum_{i,j} n p_j x_{ij}^* = \textrm{OPT}_{\text{LP}}$.

\paragraph{Upper Bound} Consider any realized graph $G$ and let $M(G)$ be the maximum bipartite matching on this graph. We can construct a corresponding expected fractional matching $\bar{x}$. Let $\bar{x}_{ij}$ be the expected number of times a request of type $t_j$ is matched to resource $v_i$ in $M(G)$, normalized by the expected number of arrivals $n p_j$. Because $M(G)$ is a valid matching for every realization, it must respect the matching constraint: (1) $\sum_j n p_j \bar{x}_{ij} \le 1$ , and (2) $\sum_i \bar{x}_{ij} \le 1$. Since $\bar{x}_{ij}$ satisfies all the constraints of the Expected Instance LP, it is a feasible solution. Because $x^*$ is the optimal solution, the objective value produced by $\bar{x}$ cannot exceed $\textrm{OPT}_{\text{LP}}$:
\begin{align*}
    \mathbb{E}[|M(G)|] = \sum_{j=1}^{m} n p_j \sum_{v_i \in \Gamma(t_j)} \bar{x}_{ij} \le \textrm{OPT}_{\text{LP}}
\end{align*}

\paragraph{Lower Bound} To prove the lower bound, we consider the following randomized routing strategy. Because the maximum offline matching $M(G)$ is optimal, its size is bounded below by the number of matches produced by any valid allocation strategy. Suppose that upon arrival, each of the $n$ requests independently draws its type from distribution $D$. If the request is of type $t_j$, it independently chooses to route a match to a compatible resource $v_i$ with probability $x_{ij}^*$. (Because $\sum_i x_{ij}^* \le 1$, this is a valid probability distribution). Let $\Lambda_i$ be the expected number of routing attempts directed at resource $v_i$. From the LP constraints, we know:
$$\Lambda_i = \sum_{j} n p_j x_{ij}^* \le 1$$

Since there are exactly $n$ independent arrivals, the probability that resource $v_i$ receives exactly zero routing attempts is:
$$\mathbb{P}(v_i \text{ is unmatched}) = \left( 1 - \frac{\sum_j n p_j x_{ij}^*}{n} \right)^n = \left( 1 - \frac{\Lambda_i}{n} \right)^n$$

A resource is successfully matched if it receives at least one routing attempt (any collisions are simply dropped). The probability that resource $v_i$ is matched is therefore $1 - (1 - \Lambda_i/n)^n$. By the standard properties of this function for $\Lambda_i \in [0, 1]$, we have:$$1 - \left( 1 - \frac{\Lambda_i}{n} \right)^n \ge \left(1 - \left(1 - \frac{1}{n}\right)^n \right) \Lambda_i \ge \left(1 - \frac{1}{e}\right) \Lambda_i$$

By linearity of expectation, we sum the probability of being matched across all resources $v_i \in V$ to bound the expected size of this matching, which in turn lower bounds the maximum matching:
  $$\mathbb{E}[|M(G)|] \ge \sum_{i \in V} \left(1 - \frac{1}{e}\right) \Lambda_i = \left(1 - \frac{1}{e}\right) \sum_{i \in V} \sum_{j=1}^m n p_j x_{ij}^* = \left(1 - \frac{1}{e}\right) \textrm{OPT}_{\text{LP}}$$
  
This concludes the proof. 

\section{Proof of \Cref{thm:approx-ratio}} \label{app:theorem-proof}

Our strategy is to specify a feasible fractional matching $f$ on the realized sparsifier $G_S$ and bound the ratio of its expected size $\E[\|f\|_1]$ to the LP objective $Z(x)$. Since any fractional matching on a bipartite graph can be rounded to an integral matching of the same size without loss, it holds that $\E[|M(G_S)|] \geq \E[\|f\|_1]$.

\begin{remark}[Micro-types Limit] \label{rem:microtypes}
To simplify the analysis, 
we assume without loss of generality that the expected arrival rates of individual types are infinitesimally small. We can obtain an equivalent model by conceptually splitting any original type (with single-draw probability $p_j$) into $m$ identical ``micro-types'', each with probability $p_j/m$, and taking the limit as $m \to \infty$.
In this limit, the probability of multiple arrivals of the exact same micro-type vanishes. Consequently, the process of drawing exactly $n$ independent arrivals is mathematically equivalent to selecting a fixed-size sample of $n$ distinct micro-types without replacement, where a specific micro-type of original type $j$ is included with probability $n (p_j/m)$.
\end{remark}

With the micro-types limit, the global arrival of $n$ requests acts as a size-$n$ \varopt\ sample over micro-types. The entire stochastic process therefore composes two negatively-associated sampling schemes: the global arrival of requests and the local selection of $k$ edges per request. Thus, the final subset of edges in $G_S$ and at each resource behaves as a unified, negatively-associated \varopt\ sample.

\subsection{Fractional Matching on the Sparsifier}

For each request $u$ of type $t_j$ that arrives, we assign each locally sampled edge $e = (u, v_i)$ its corresponding inverse probability weight $w_e$:
\begin{itemize}
    \item \emph{Light Edges ($e \in \mathcal{E}_L$):} When $x_{ij} \leq 1/k$, $\pi_{ij} \geq k x_{ij}$, which ensures $w_e = x_{ij} / \pi_{ij} \leq 1/k$.
    \item \emph{Heavy Edges ($e \in \mathcal{E}_H$):} When $x_{ij} > 1/k$, the \varopt~ thresholding guarantees $\pi_{ij} = 1$ and therefore $w_e = x_{ij}$. The edge is deterministically included whenever the type arrives.
\end{itemize}

By the deterministic weight-preservation property of \varopt (\cref{eq:varopt_sum}), the sum of weights leaving any arriving request exactly matches the fractional capacity of the type: $\sum_{i: (u, v_i) \in E_S} w_e = \sum_{i} x_{ij} \leq 1$. Therefore, the capacity constraints are naturally satisfied at the request nodes.

However, collisions may cause the total arriving weight at a resource $v_i$, denoted $Y_i = \sum_{e \in E_S \cap \delta(v_i)} w_e$, to exceed $1$. To obtain a strictly feasible fractional matching $f$, we uniformly scale down the edges incident to any overloaded resource:
\[ f_e := \frac{w_e}{\max(1, Y_i)} \quad \text{for edge } e = (u, v_i) \]
It is easy to verify that $f$ is a valid fractional matching on $G_S$. Since $\E[\sum_i Y_i] = \E[\|w\|_1] = Z(x)$, the expected size of the matching is the total expected demand minus the expected \emph{excess demand} that is scaled away:
\begin{equation} \label{eq:excess-gap}
\E[\|f\|_1] = Z(x) - \sum_{i \in V} \E[\max(0, Y_i - 1)]\ .
\end{equation}

\subsection{Global Excess of Compound Poisson Loads} \label{sec:CompundPoissonModel}

We proceed to bound the expected excess at resources. We assume independent arrivals at each resource, since negative associations only decrease the excess \citep{Shao:2000}. In the micro-type limit, the load is a compound Poisson distribution. 
We formulate the maximum total expected excess as an optimization problem. We then identify load reassignments that can only increase this expected excess to arrive at a certain normalized form. We express upper bounds on the expected excess at a resource dependent on its load distribution. We then optimize over this upper bound to derive an upper bound on the total expected excess.

\paragraph{Global demand.}
Let $\rho$ be an intensity measure on $(0,1]$. For any measurable set
$S\subseteq(0,1]$, the quantity
\[
\rho(S)
\]
is the expected number of arriving items whose sizes lie in $S$.

The total expected load is
\[
Z := \int_0^1 w\,\rho(dw) < \infty .
\]

\paragraph{Resources and feasible partitions.}
Let $V$ be an index set of bins. A feasible allocation is a collection of
intensity measures
\[
\{\rho_i\}_{i\in V}
\]
on $(0,1]$ such that:

\begin{enumerate}
    \item \textbf{Mass conservation:}
    \[
    \sum_{i\in V}\rho_i = \rho .
    \]

    \item \textbf{Capacity constraints:}
    \[
    \Lambda_i := \int_0^1 w\,\rho_i(dw) \le 1
    \qquad \forall i\in V.
    \]
\end{enumerate}

\paragraph{Local load model.}
For each bin $i\in V$, let
\[
Y_i\sim \mathrm{CP}(\rho_i)
\]
be the compound Poisson load with intensity measure $\rho_i$.

Equivalently, if
\[
\lambda_i:=\rho_i((0,1])
\]
is finite, then
\[
Y_i=\sum_{k=1}^{N_i} W_{i,k},
\]
where
\[
N_i\sim\mathrm{Poisson}(\lambda_i),
\]
and, conditioned on $N_i$, the sizes $W_{i,k}$ are i.i.d. with distribution
\[
\Pr[W_{i,k}\in S]=\frac{\rho_i(S)}{\lambda_i}.
\]
The mean load satisfies
\[
\mathbb E[Y_i]=\Lambda_i.
\]

\paragraph{Expected excess functional.}
Define the expected excess of bin $i$ by
\[
\mathcal E[\rho_i]
:=
\mathbb E[(Y_i-1)_+],
\qquad
(x)_+:=\max\{x,0\}.
\]

\paragraph{Optimization problem.}
The quantity of interest is
\[
\mathrm{OPT}(\rho)
:=
\sup_{\{\rho_i\}_{i\in V}}
\sum_{i\in V}\mathcal E[\rho_i],
\]
where the supremum is over all feasible partitions satisfying mass
conservation and the capacity constraints above.

\paragraph{Finite-support convention.}
In the remainder of this section, unless stated otherwise, we assume that
the global intensity measure is supported on a finite set of sizes
\[
\mathcal W\subset(0,1].
\]
We write $\rho_i(w)$ for the Poisson rate of size-$w$ items assigned to bin
$i$, so that
\[
\rho_i(dw)=\sum_{w\in\mathcal W}\rho_i(w)\delta_w(dw).
\]
Thus integrals are identified with the corresponding finite sums; for
example,
\[
\Lambda_i
=
\int_0^1 x\,\rho_i(dx)
=
\sum_{w\in\mathcal W} w\,\rho_i(w).
\]

\subsection{Pairwise size-wise optimization process}

We next describe a local improvement procedure for a feasible partition in
the finite-support setting. The procedure operates on two residual bins and
one size whose intensity is currently assigned to both bins. Keeping all
other sizes fixed, we optimize the allocation of this one size between the
two bins. The endpoint property below shows that this one-dimensional
optimization admits an optimal endpoint solution: either one of the two bins
becomes full, or all intensity of this size assigned to the pair is moved to
a single bin. Thus the update can be performed without decreasing the total
expected excess.

Iterating such pairwise size-wise endpoint updates yields a useful normal
form. After removing bins that become full, the remaining residual bins can
be chosen to have disjoint size supports: no size with positive residual
intensity is split across two distinct residual bins.

\begin{lemma}[Endpoint property with arbitrary background loads]
Let $Y_1$ and $Y_2$ be independent nonnegative random variables with
$\mathbb E[Y_1]+\mathbb E[Y_2]<\infty$. Let $w>0$ and $R\ge 0$.
For $a\in[0,R]$, define
\[
G(a)
:=
\mathbb E\bigl[(Y_1+wN_a-1)_+\bigr]
+
\mathbb E\bigl[(Y_2+wN_{R-a}-1)_+\bigr],
\]
where $N_a\sim \mathrm{Poisson}(a)$ and
$N_{R-a}\sim \mathrm{Poisson}(R-a)$ are independent of each other and of
$Y_1,Y_2$. Then $G$ is convex on $[0,R]$. Consequently, on every compact
interval $I\subseteq[0,R]$,
\[
\max_{a\in I}G(a)
=
\max\{G(\ell),G(u)\},
\qquad I=[\ell,u].
\]
\end{lemma}

\begin{proof}
Let
\[
\psi(x):=(x-1)_+.
\]
We first show that
\[
f_1(a):=\mathbb E[\psi(Y_1+wN_a)]
\]
is convex in $a$.

Condition on $Y_1=y$. Define
\[
\phi_y(n):=\psi(y+wn)=(y+wn-1)_+.
\]
Since $\psi$ is convex and $n\mapsto y+wn$ is affine,
$\phi_y$ is convex as a function of the real variable $n$.

Now define
\[
g_y(a):=\mathbb E[\phi_y(N_a)].
\]
For any function
$\phi$ of at most linear growth, the standard Poisson differentiation
identity gives
\[
\frac{d}{da}\mathbb E[\phi(N_a)]
=
\mathbb E[\phi(N_a+1)-\phi(N_a)].
\]
Applying this identity twice gives
\[
g_y''(a)
=
\mathbb E\bigl[
\phi_y(N_a+2)-2\phi_y(N_a+1)+\phi_y(N_a)
\bigr].
\]
Because $\phi_y$ is convex,
\[
\phi_y(n+2)-2\phi_y(n+1)+\phi_y(n)\ge 0
\qquad \forall n\ge 0.
\]
Hence
\[
g_y''(a)\ge 0,
\]
so $g_y$ is convex in $a$ for every fixed $y$.

Therefore
\[
f_1(a)=\mathbb E_{Y_1}[g_{Y_1}(a)]
\]
is convex, since a nonnegative weighted average of convex functions is
convex.

The same argument shows that
\[
f_2(u):=\mathbb E[\psi(Y_2+wN_u)]
\]
is convex in $u$. Since $u=R-a$ is an affine function of $a$, the function
\[
a\mapsto f_2(R-a)
\]
is convex. Thus
\[
G(a)=f_1(a)+f_2(R-a)
\]
is convex on $[0,R]$.

Finally, a convex function on a compact interval attains its maximum at
one of the two endpoints. Therefore, for any compact feasible interval
$I=[\ell,u]\subseteq[0,R]$,
\[
\max_{a\in I}G(a)=\max\{G(\ell),G(u)\}.
\]
\end{proof}

\begin{definition}[Pairwise size-wise endpoint process]
Let $\{\rho_i\}_{i\in V}$ be a feasible partition of the global intensity
measure $\rho$, and define
\[
\Lambda_i := \int_0^1 x\,\rho_i(dx) \le 1.
\]
A bin $i$ is called \emph{full} if $\Lambda_i=1$ and \emph{residual} if
$\Lambda_i<1$.

The process repeatedly performs the following operation.

\begin{enumerate}
    \item Choose two residual bins $i,j$ and a size $w\in\mathcal W$ such that
    \[
    \rho_i(w)>0
    \qquad\text{and}\qquad
    \rho_j(w)>0.
    \]
    Let
    \[
    R := \rho_i(w)+\rho_j(w)
    \]
    denote the total size-$w$ rate assigned to this pair.

    \item Freeze all rates in bins $i,j$ at sizes $w'\neq w$, and redistribute
    only the size-$w$ rate $R$: assign rate $a$ to bin $i$ and rate $R-a$ to
    bin $j$.

    \item Choose $a$ from the feasible interval so as to maximize the two-bin
    contribution
    \[
    \mathbb E\bigl[(Y_i^{\setminus w}+wN_a-1)_+\bigr]
    +
    \mathbb E\bigl[(Y_j^{\setminus w}+wN_{R-a}-1)_+\bigr],
    \]
    where $Y_i^{\setminus w}$ and $Y_j^{\setminus w}$ are the background loads
    generated by all sizes other than $w$. Among maximizers, choose an
    endpoint maximizer; by the endpoint property, such a maximizer exists.
    Thus the update can be chosen so that either one of the two bins becomes
    full, or all size-$w$ rate in the pair is assigned to a single bin.

    \item If a bin becomes full, remove it from subsequent operations.
\end{enumerate}

The process terminates when no two residual bins carry positive rate at a
common size $w\in\mathcal W$.
\end{definition}

\begin{corollary}[Endpoint form of a pairwise update]
Consider a pairwise update on residual bins $i,j$ and size $w$, and let
\[
R:=\rho_i(w)+\rho_j(w).
\]
Let
\[
\Lambda_i^{\setminus w}:=\Lambda_i-w\rho_i(w),
\qquad
\Lambda_j^{\setminus w}:=\Lambda_j-w\rho_j(w).
\]
The feasible interval for the rate $a$ assigned to bin $i$ is
\[
I_{i,j,w}
=
\left\{
a\in[0,R]:
\Lambda_i^{\setminus w}+wa\le 1,\;
\Lambda_j^{\setminus w}+w(R-a)\le 1
\right\}.
\]
There is an optimal choice of $a\in I_{i,j,w}$ at an endpoint of this
interval. Consequently, after such an endpoint update, at least one of the
following holds:
\[
\rho_i(w)=0,\qquad
\rho_j(w)=0,\qquad
\Lambda_i=1,\qquad
\Lambda_j=1.
\]
\end{corollary}

\begin{proof}
Let $Y_i^{\setminus w}$ and $Y_j^{\setminus w}$ denote the loads in the two
bins generated by all sizes other than $w$. If rate $a$ of size $w$ is
assigned to bin $i$, and rate $R-a$ to bin $j$, the two-bin objective is
\[
G(a)
=
\mathbb E\bigl[(Y_i^{\setminus w}+wN_a-1)_+\bigr]
+
\mathbb E\bigl[(Y_j^{\setminus w}+wN_{R-a}-1)_+\bigr].
\]
By the endpoint property, $G$ is convex on $[0,R]$, and hence its maximum
over the compact feasible interval $I_{i,j,w}$ is attained at an endpoint.

The endpoints of $I_{i,j,w}$ arise when one of the defining inequalities is
tight:
\[
a=0,\qquad a=R,\qquad
\Lambda_i^{\setminus w}+wa=1,\qquad
\Lambda_j^{\setminus w}+w(R-a)=1.
\]
These four cases are exactly
\[
\rho_i(w)=0,\qquad
\rho_j(w)=0,\qquad
\Lambda_i=1,\qquad
\Lambda_j=1,
\]
respectively.
\end{proof}

\begin{lemma}[Structure of terminal points]
Consider a terminal point of the pairwise size-wise endpoint process, and let
\[
\mathcal R := \{i\in V : \Lambda_i<1\}
\]
denote the set of residual bins. Then, for every size $w\in\mathcal W$, at
most one residual bin carries positive size-$w$ intensity:
\[
\bigl|\{i\in\mathcal R : \rho_i(w)>0\}\bigr| \le 1.
\]
Equivalently, for any two distinct residual bins $i,j\in\mathcal R$,
\[
\operatorname{supp}(\rho_i)\cap \operatorname{supp}(\rho_j)=\emptyset .
\]
\end{lemma}

\begin{proof}
Suppose, toward a contradiction, that there exist two distinct residual bins
$i,j\in\mathcal R$ and a size $w\in\mathcal W$ such that
\[
\rho_i(w)>0
\qquad\text{and}\qquad
\rho_j(w)>0.
\]
Then the pairwise size-wise endpoint update is applicable to the pair
$(i,j)$ and size $w$. By the endpoint form of a pairwise update, the
size-$w$ intensity assigned to the pair can be redistributed, without
decreasing the objective, so that at least one of the following holds:
\[
\rho_i(w)=0,\qquad
\rho_j(w)=0,\qquad
\Lambda_i=1,\qquad
\Lambda_j=1.
\]
Thus the update either removes size $w$ from one of the two residual bins,
or makes one of the two bins full and hence removes it from the residual
set. In either case, it eliminates this shared occurrence of size $w$ among
residual bins, while not decreasing the total expected excess.

This contradicts terminality of the process, which stops only when no two
residual bins share positive intensity at any size. Therefore, for every
$w\in\mathcal W$, at most one residual bin can carry positive size-$w$
intensity.
\end{proof}

Consequently, without decreasing the objective, we may restrict attention to
terminal partitions in which the residual bins have disjoint size supports.

\subsection{Monotonicity under upward size replacements}

The following lemma formalizes the principle that, for a fixed expected
load, replacing smaller jumps by larger jumps can only increase expected
excess. Equivalently, making the load distribution more ``chunky'' increases
all convex cost functionals.

\begin{lemma}[Monotonicity under upward size replacements]\label{lem:monotonicity}
Let $\mu$ be an intensity measure on $(0,1]$, and let
\[
Y\sim \mathrm{CP}(\mu).
\]
Fix sizes $0<s<s'\le 1$ and a rate $a>0$ such that $\mu(\{s\})\ge a$.
Define
\[
\mu'
:=
\mu - a\delta_s + \frac{as}{s'}\delta_{s'}.
\]
That is, rate $a$ of size-$s$ items is replaced by rate $as/s'$ of
size-$s'$ items.

Then the expected load is unchanged:
\[
\int_0^1 x\,\mu'(dx)
=
\int_0^1 x\,\mu(dx).
\]
Moreover, if
\[
Y'\sim \mathrm{CP}(\mu'),
\]
then
\[
Y \preceq_{\mathrm{cx}} Y',
\]
that is, for every convex function $\phi:\mathbb R_{\ge 0}\to\mathbb R$ for
which the expectations are finite,
\[
\mathbb E[\phi(Y)]\le \mathbb E[\phi(Y')].
\]
In particular,
\[
\mathbb E[(Y-1)_+]\le \mathbb E[(Y'-1)_+].
\]
\end{lemma}

\begin{proof}
The equality of expected loads is immediate:
\[
as=\frac{as}{s'}\,s'.
\]

It remains to prove the convex-order claim. Decompose the original and
modified compound Poisson loads as
\[
Y=B+X,
\qquad
Y'=B+X',
\]
where $B$ is the independent load generated by all intensity except the
modified atom, and
\[
X=sN_a,
\qquad
X'=s'N_{as/s'}.
\]
Here
\[
N_a\sim \mathrm{Poisson}(a),
\qquad
N_{as/s'}\sim \mathrm{Poisson}(as/s').
\]
Both $X$ and $X'$ have mean $as$.

We claim that
\[
X\preceq_{\mathrm{cx}} X'.
\]
We use the following standard convex-order criterion for compound Poisson
random variables: if two compound Poisson intensity measures $\nu$ and
$\nu'$ have the same first moment and satisfy
\[
\int_0^1 (x-t)_+\,\nu(dx)
\le
\int_0^1 (x-t)_+\,\nu'(dx)
\qquad \forall t\ge 0,
\]
then
\[
\mathrm{CP}(\nu)\preceq_{\mathrm{cx}}\mathrm{CP}(\nu').
\]

Apply this criterion to
\[
\nu=a\delta_s,
\qquad
\nu'=\frac{as}{s'}\delta_{s'}.
\]
The first moments agree:
\[
\int_0^1 x\,\nu(dx)=as=\int_0^1 x\,\nu'(dx).
\]
For every $t\ge 0$,
\[
\int_0^1 (x-t)_+\,\nu(dx)
=
a(s-t)_+,
\]
whereas
\[
\int_0^1 (x-t)_+\,\nu'(dx)
=
\frac{as}{s'}(s'-t)_+.
\]
We verify that
\[
a(s-t)_+
\le
\frac{as}{s'}(s'-t)_+
\qquad\forall t\ge 0.
\]
If $t\ge s'$, both sides are zero. If $s\le t<s'$, the left-hand side is
zero and the right-hand side is nonnegative. Finally, if $0\le t<s$, the
claim is equivalent to
\[
s-t \le s-\frac{s}{s'}t,
\]
or
\[
t\ge \frac{s}{s'}t,
\]
which holds since $s<s'$.

Therefore
\[
X\preceq_{\mathrm{cx}} X'.
\]
Convex order is preserved under addition of an independent random variable,
so
\[
B+X\preceq_{\mathrm{cx}}B+X'.
\]
Hence
\[
Y\preceq_{\mathrm{cx}}Y'.
\]
Taking $\phi(x)=(x-1)_+$ gives
\[
\mathbb E[(Y-1)_+]\le \mathbb E[(Y'-1)_+].
\]
\end{proof}

\begin{corollary}[Iterated upward replacements] \label{cor:monotonicity}
If $\mu'$ is obtained from $\mu$ by a finite sequence of replacements of the
form
\[
a\delta_s
\longmapsto
\frac{as}{s'}\delta_{s'}
\qquad\text{with } 0<s\le s'\le 1,
\]
then
\[
\mathrm{CP}(\mu)\preceq_{\mathrm{cx}}\mathrm{CP}(\mu').
\]
Consequently,
\[
\mathbb E[(Y-1)_+]\le \mathbb E[(Y'-1)_+],
\]
where
\[
Y\sim\mathrm{CP}(\mu),
\qquad
Y'\sim\mathrm{CP}(\mu').
\]
\end{corollary}

\begin{proof}
Apply the preceding lemma to each replacement in the sequence and use
transitivity of convex order.
\end{proof}

\subsection{Upper bounds on expected excess}

We first derive intensity-dependent expected excess upper bounds for a single bin. We then show that a balanced partition within the loads maximizes the global upper bound.

\begin{lemma}[Variance upper bound for expected excess]\label{lem:singlebinvar}
Let $\rho$ be an intensity measure on $(0,1]$, and let
\[
Y \sim \mathrm{CP}(\rho)
\]
be the corresponding compound Poisson load. Suppose that
\[
\mathbb E[Y]
=
\int_0^1 w\,\rho(dw)
\le 1.
\]
Then
\[
\mathcal E[\rho]
:=
\mathbb E[(Y-1)_+]
\le
\frac{1}{2}\sqrt{\operatorname{Var}(Y)}
=
\frac{1}{2}\sqrt{\int_0^1 w^2\,\rho(dw)}.
\]
\end{lemma}

\begin{proof}
Let
\[
m:=\mathbb E[Y]\le 1.
\]
Since
\[
(Y-1)_+ \le (Y-m)_+,
\]
we have
\[
\mathcal E[\rho]
=
\mathbb E[(Y-1)_+]
\le
\mathbb E[(Y-m)_+].
\]

Because $\mathbb E[Y-m]=0$, the positive and negative parts have equal expectation:
\[
\mathbb E[(Y-m)_+]
=
\mathbb E[(m-Y)_+].
\]
Therefore,
\[
\mathbb E[(Y-m)_+]
=
\frac12 \mathbb E[|Y-m|].
\]
By Cauchy--Schwarz,
\[
\mathbb E[|Y-m|]
\le
\sqrt{\mathbb E[(Y-m)^2]}
=
\sqrt{\operatorname{Var}(Y)}.
\]
Hence
\[
\mathcal E[\rho]
\le
\frac12\sqrt{\operatorname{Var}(Y)}.
\]

Finally, for a compound Poisson random variable with intensity measure
$\rho$ on jump sizes,
\[
\operatorname{Var}(Y)
=
\int_0^1 w^2\,\rho(dw).
\]
Thus
\[
\mathcal E[\rho]
\le
\frac12
\sqrt{
\int_0^1 w^2\,\rho(dw)
}.
\]
\end{proof}

We now derive a specialized bound for the case where all mass is on sizes $1$ and $\tau<1$. 

\begin{lemma}[$\eta$-dependent excess bound for two sizes]\label{lem:singlebin2s}
Fix $0<\tau<1$ and $\eta\in[0,1]$. Consider a load-one compound Poisson bin
whose size-$1$ intensity is $\eta$ and whose size-$\tau$ intensity is
$(1-\eta)/\tau$. Equivalently,
\[
Y
:=
N_\eta+\tau N_{(1-\eta)/\tau},
\]
where
\[
N_\eta\sim \mathrm{Poisson}(\eta),
\qquad
N_{(1-\eta)/\tau}\sim \mathrm{Poisson}\!\left(\frac{1-\eta}{\tau}\right),
\]
independently. 
Then
\[
\mathbb E[Y]=1
\]
and
\[
\mathbb E[(Y-1)_+]
=
e^{-\eta}
\mathbb E\left[
\left(1-\tau N_{(1-\eta)/\tau}\right)_+
\right].
\]
Consequently,
\[
\mathbb E[(Y-1)_+]
\le
e^{-\eta}
\left(
\eta+\frac12\sqrt{\tau(1-\eta)}
\right).
\]
\end{lemma}

\begin{proof}
First,
\[
\mathbb E[Y]
=
\eta+\tau\cdot \frac{1-\eta}{\tau}
=
1.
\]
Therefore,
\[
\mathbb E[(Y-1)_+]
=
\mathbb E[(1-Y)_+].
\]

Let
\[
M:=N_{(1-\eta)/\tau}.
\]
If $N_\eta\ge 1$, then
\[
Y=N_\eta+\tau M\ge 1,
\]
and hence
\[
(1-Y)_+=0.
\]
Thus the deficit can occur only on the event $\{N_\eta=0\}$. Since
$N_\eta$ and $M$ are independent,
\[
\mathbb E[(1-Y)_+]
=
\Pr[N_\eta=0]\,
\mathbb E[(1-\tau M)_+]
=
e^{-\eta}
\mathbb E[(1-\tau M)_+].
\]
This proves the exact identity.

It remains to bound the last expectation. Write
\[
1-\tau M
=
\eta+\bigl((1-\eta)-\tau M\bigr).
\]
Therefore,
\[
(1-\tau M)_+
\le
\eta+\bigl((1-\eta)-\tau M\bigr)_+.
\]
Now
\[
\mathbb E[\tau M]=1-\eta
\]
and
\[
\operatorname{Var}(\tau M)
=
\tau^2\operatorname{Var}(M)
=
\tau^2\cdot \frac{1-\eta}{\tau}
=
\tau(1-\eta).
\]
For any mean-zero random variable $X$,
\[
\mathbb E[X_+]
=
\frac12 \mathbb E[|X|]
\le
\frac12\sqrt{\mathbb E[X^2]}.
\]
Apply this to
\[
X:=(1-\eta)-\tau M,
\]
which has mean zero. Then
\[
\mathbb E\bigl[((1-\eta)-\tau M)_+\bigr]
\le
\frac12\sqrt{\tau(1-\eta)}.
\]
Hence
\[
\mathbb E[(1-\tau M)_+]
\le
\eta+\frac12\sqrt{\tau(1-\eta)}.
\]
Multiplying by $e^{-\eta}$ gives
\[
\mathbb E[(Y-1)_+]
\le
e^{-\eta}
\left(
\eta+\frac12\sqrt{\tau(1-\eta)}
\right).
\]
\end{proof}

\begin{corollary}
    For a single bin and the two size setting, the expected excess is bounded by
\[
\mathbb E[(Y-1)_+]
\le
\min\left\{
\frac12\sqrt{\eta+\tau(1-\eta)},
\;
e^{-\eta}
\left(
\eta+\frac12\sqrt{\tau(1-\eta)}
\right)
\right\}.
\]
\end{corollary}
\begin{proof}
    Combining \cref{lem:singlebinvar} and \cref{lem:singlebin2s}.
\end{proof}

\begin{lemma}[Balanced allocation maximizes the two-size upper bound]\label{lem:balancedbound}
Fix $0<\tau<1$ and let there be $m$ full bins, each of expected load $1$.
Assume that the only sizes are $1$ and $\tau$, and that the total load
contributed by size-$1$ items across these $m$ bins is $m\eta$, where
$\eta\in[0,1]$. Equivalently, if $\eta_i$ denotes the size-$1$ load in bin
$i$, then
\[
0\le \eta_i\le 1,
\qquad
\frac1m\sum_{i=1}^m \eta_i=\eta,
\]
and the size-$\tau$ load in bin $i$ is $1-\eta_i$.

For a full bin with size-$1$ load $\eta_i$, let $Y_i$ denote its compound
Poisson load. Then
\[
\mathbb E[(Y_i-1)_+]
\le
U_\tau(\eta_i),
\]
where
\[
U_\tau(x)
:=
\min\left\{
\frac12\sqrt{x+\tau(1-x)},
\;
e^{-x}\left(x+\frac12\sqrt{\tau(1-x)}\right)
\right\}.
\]
Moreover,
\[
\sum_{i=1}^m \mathbb E[(Y_i-1)_+]
\le
m\,U_\tau(\eta).
\]
Thus, for this upper bound, the worst case is attained by balancing the
size-$1$ load across the full bins:
\[
\eta_i=\eta
\qquad\forall i.
\]
\end{lemma}
\begin{proof}
For bin $i$, the load has the form
\[
Y_i=N_{\eta_i}+\tau N_{(1-\eta_i)/\tau},
\]
where the two Poisson variables are independent. Indeed,
\[
\mathbb E[Y_i]
=
\eta_i+\tau\cdot \frac{1-\eta_i}{\tau}
=
1.
\]

From the variance bound,
\[
\mathbb E[(Y_i-1)_+]
\le
\frac12\sqrt{\operatorname{Var}(Y_i)}.
\]
Here
\[
\operatorname{Var}(Y_i)
=
\eta_i+\tau(1-\eta_i),
\]
so
\[
\mathbb E[(Y_i-1)_+]
\le
B_1(\eta_i),
\qquad
B_1(x):=\frac12\sqrt{x+\tau(1-x)}.
\]

The two-size deficit bound gives
\[
\mathbb E[(Y_i-1)_+]
\le
B_2(\eta_i),
\]
where
\[
B_2(x)
:=
e^{-x}\left(x+\frac12\sqrt{\tau(1-x)}\right).
\]
Therefore,
\[
\mathbb E[(Y_i-1)_+]
\le
\min\{B_1(\eta_i),B_2(\eta_i)\}
=
U_\tau(\eta_i).
\]

We now sum these bounds. Since
\[
\min\{B_1(\eta_i),B_2(\eta_i)\}\le B_1(\eta_i)
\]
for every $i$, we have
\[
\sum_{i=1}^m U_\tau(\eta_i)
\le
\sum_{i=1}^m B_1(\eta_i).
\]
The function $B_1$ is concave on $[0,1]$ (see \cref{claim:concavity}), so Jensen's inequality gives
\[
\sum_{i=1}^m B_1(\eta_i)
\le
m B_1\!\left(\frac1m\sum_{i=1}^m \eta_i\right)
=
mB_1(\eta).
\]
Similarly,
\[
\sum_{i=1}^m U_\tau(\eta_i)
\le
\sum_{i=1}^m B_2(\eta_i).
\]
The function $B_2$ is concave on $[0,1]$ (see \cref{claim:concavity}). Hence Jensen's inequality gives
\[
\sum_{i=1}^m B_2(\eta_i)
\le
m B_2(\eta).
\]
Combining the two estimates,
\[
\sum_{i=1}^m U_\tau(\eta_i)
\le
\min\{mB_1(\eta),mB_2(\eta)\}
=
mU_\tau(\eta).
\]
Thus
\[
\sum_{i=1}^m \mathbb E[(Y_i-1)_+]
\le
mU_\tau(\eta).
\]
This proves the claim.
\end{proof}

\begin{claim} [Concavity of Upper bound functions] \label{claim:concavity}
The functions
\[
B_1(x)=\frac12\sqrt{x+\tau(1-x)}
\]
and
\[
B_2(x)=e^{-x}\left(x+\frac12\sqrt{\tau(1-x)}\right)
\]
are concave on $[0,1]$.
\end{claim}

\begin{proof}
The concavity of $B_1$ is immediate because it is the square root of an
affine function.

For $B_2$, write
\[
q:=\sqrt{\tau},
\qquad
c:=\frac q2\le \frac12,
\]
so
\[
B_2(x)=e^{-x}\bigl(x+c\sqrt{1-x}\bigr).
\]
For $x\in[0,1)$, set
\[
u:=\sqrt{1-x}.
\]
A direct calculation gives
\[
B_2''(x)
=
\frac{e^{-x}}{4(1-x)^2}
\left[
c(4u^5+4u^3-u)-4u^4(1+u^2)
\right].
\]
We show that the bracketed term is nonpositive. If
\[
4u^5+4u^3-u\le 0,
\]
this is immediate. Otherwise, using $c\le 1/2$, it is enough to show
\[
\frac12(4u^5+4u^3-u)
\le
4u^4(1+u^2).
\]
Equivalently,
\[
8u^5-4u^4+8u^3-4u^2+1\ge 0.
\]
If $u\ge 1/2$, this is immediate since
\[
8u^5-4u^4+8u^3-4u^2
=
4u^4(2u-1)+4u^2(2u-1)\ge 0.
\]
If $0\le u\le 1/2$, then
\[
8u^5-4u^4+8u^3-4u^2+1
=
1-4u^2(1+u^2)(1-2u).
\]
The function $u^2(1-2u)$ is maximized on $[0,1/2]$ at $u=1/3$, with value
$1/27$, and $1+u^2\le 5/4$ on this interval. Therefore
\[
4u^2(1+u^2)(1-2u)
\le
4\cdot \frac54 \cdot \frac1{27}
=
\frac5{27}<1.
\]
Hence the expression is nonnegative in all cases. Thus $B_2''(x)\le 0$ on
$[0,1)$, and $B_2$ is concave on $[0,1]$.
\end{proof}

\subsection{Proof of the approximation theorem}

We can now complete the proof of the global approximation guarantee.

\begin{proof}[Proof of Theorem~\ref{thm:approx-ratio}]
Fix a feasible fractional solution $x$, and let
\[
Z:=Z(x)=Z_H+Z_L,
\]
where $Z_H$ and $Z_L$ denote the heavy and light components of the objective,
respectively. If $Z=0$, the claim is trivial, so assume $Z>0$.

Recall from \eqref{eq:excess-gap} that the expected size of
the fractional matching constructed on the sparsifier satisfies
\[
\mathbb E[\|f\|_1]
=
Z
-
\sum_{i\in V}\mathbb E[(Y_i-1)_+],
\]
where $Y_i$ is the total sampled load arriving at resource $i$. Thus it
suffices to upper bound the total expected excess.

We first transform the load profile in a way that can only increase the
expected excess. Let
\[
\tau:=\frac1k.
\]
By the monotonicity lemma and its iterated form
(\Cref{lem:monotonicity,cor:monotonicity}), we may replace every item size
in $(\tau,1]$ by size $1$, preserving its expected load, and every item size
in $(0,\tau]$ by size $\tau$, again preserving its expected load. This
produces a two-size load profile supported on $\{1,\tau\}$, with total
size-$1$ load $Z_H$ and total size-$\tau$ load $Z_L$. Since this
transformation only increases the compound Poisson loads in convex order,
and since $y\mapsto (y-1)_+$ is convex, it can only increase the total
expected excess.

We now apply the pairwise size-wise endpoint process to this two-size
instance. By the terminal-point normal form, after removing full bins, no
size is shared by two residual bins. Since there are only two sizes, at
termination there are at most two residual bins: one carrying only size-$1$
load and one carrying only size-$\tau$ load.

If the combined expected load of these residual bins is at most $1$, we
merge them into a single bin. This cannot decrease the expected excess:
for independent nonnegative loads $X$ and $Y$,
\[
(X+Y-1)_+
\ge
(X-1)_+ + (Y-1)_+ .
\]
Thus, after this possible merge, the total residual load is either contained
in one bin of load at most $1$, or it is contained in two residual bins whose
combined load exceeds $1$. In either case, the total unused capacity among
the residual bins is at most $1$.

We next fill the remaining residual capacity by adding extra independent
load, using the same heavy-light proportions as the original total load.
Specifically, set
\[
\eta:=\frac{Z_H}{Z}.
\]
Whenever we add an amount $c$ of expected load, we add $\eta c$ as size-$1$
load and $(1-\eta)c$ as size-$\tau$ load. Adding independent nonnegative
load can only increase $(Y_i-1)_+$ pointwise, and hence can only increase
the total expected excess.

After this augmentation, all active bins are full. Let $m$ be the number of
full bins. Since the augmentation adds total load at most $1$, we have
\[
Z\le m\le Z+1.
\]
Moreover, by construction, the total size-$1$ load across the $m$ full bins
is
\[
Z_H+\eta(m-Z)=\eta m,
\]
and the total size-$\tau$ load is
\[
Z_L+(1-\eta)(m-Z)=(1-\eta)m.
\]
Thus the average size-$1$ load per full bin is exactly $\eta$.

By \Cref{lem:balancedbound}, the total expected excess over these $m$ full
bins is at most
\[
m\,U(\eta,\tau),
\]
where
\[
U(\eta,\tau)
:=
\min\left\{
\frac12\sqrt{\eta+\tau(1-\eta)},
\;
e^{-\eta}
\left(
\eta+\frac12\sqrt{\tau(1-\eta)}
\right)
\right\}.
\]
Using $m\le Z+1$, we obtain
\[
\sum_{i\in V}\mathbb E[(Y_i-1)_+]
\le
(Z+1)U(\eta,\tau).
\]
Since $U(\eta,\tau)\le 1$, this implies the slightly cleaner bound
\[
\sum_{i\in V}\mathbb E[(Y_i-1)_+]
\le
ZU(\eta,\tau)+1.
\]
Therefore,
\[
\mathbb E[\|f\|_1]
\ge
Z-ZU(\eta,\tau)-1
=
Z\bigl(1-U(\eta,\tau)\bigr)-1.
\]

Finally, substituting
\[
\eta=\frac{Z_H}{Z},
\qquad
1-\eta=\frac{Z_L}{Z},
\qquad
\tau=\frac1k,
\]
gives
\[
\mathbb E[\|f\|_1]
\ge
Z\left(
1-
\min\left\{
\frac12\sqrt{\frac{Z_H}{Z}+\frac1k\frac{Z_L}{Z}},
\;
e^{-Z_H/Z}
\left(
\frac{Z_H}{Z}
+
\frac12\sqrt{\frac1k\frac{Z_L}{Z}}
\right)
\right\}
\right)
-1.
\]
Since the maximum matching on the sparsifier has size at least the value of
any feasible fractional matching,
\[
\mathbb E[|M(G_S)|]\ge \mathbb E[\|f\|_1],
\]
which proves the theorem.
\end{proof}

\section{Implementation and Pseudocodes}

\subsection{Graph Construction} \label{sec:graph_construction}

The pseudocode for the graph construction is given in \Cref{alg:graph_construction}. 
We specifically construct balanced graphs where supply equals demand (equal number of requests and resources). We target this balanced regime because it is the point of matching market efficiency. Structural imbalances artificially simplify the allocation task; when one side of the market is vastly over-provisioned with choices, the performance of even naive algorithms degenerates to the same high baseline. Operating precisely where supply equals demand provides the necessary resolution to differentiate  sparsification strategies.

\begin{algorithm2e}[H]
\caption{Graph Construction at Time $t$ for NYC Yellow Taxis Dataset}
\label{alg:graph_construction}
\KwIn{Time interval $t$, Trip dataset $\mathcal{D}$, Set of all spatial zones $\mathcal{Z}$, Spatial adjacency matrix $A$}
\KwOut{Realized bipartite graph $G_t = (V_t \cup U_t, E_t)$}
\BlankLine

\tcp{1. Time Windows \& Data Aggregation}
$\mathcal{D}_{\text{drop}} \gets$ All drop-off events from $\mathcal{D}$ in window $[t - 5\text{m}, t)$\;
$\mathcal{D}_{\text{pick}} \gets$ All pick-up events from $\mathcal{D}$ in window $[t, t + 5\text{m})$\;
$n \gets |\mathcal{D}_{\text{drop}}|$ \tcp*{Define total graph size based on offline supply}

\For{each zone $z \in \mathcal{Z}$}{
    $p_{\text{car}}(z) \gets \frac{\text{count}(z \text{ in } \mathcal{D}_{\text{drop}})}{n}$\;
    $p_{\text{rider}}(z) \gets \frac{\text{count}(z \text{ in } \mathcal{D}_{\text{pick}})}{|\mathcal{D}_{\text{pick}}|}\;$
}
\BlankLine

\tcp{2. Proportional Node Sampling}
$V_t \gets \emptyset$ (Offline Cars), $U_t \gets \emptyset$ (Online Riders)\;
\For{$i = 1, \dots, n$}{
    Sample zone $z_v \sim p_{\text{car}}$ (with replacement)\;
    Create car node $v_i$ with location $z_v$, add to $V_t$\;
    Sample zone $z_u \sim p_{\text{rider}}$ (with replacement)\;
    Create rider node $u_i$ with location $z_u$, add to $U_t$\;
}
\BlankLine

\tcp{3. Edge Creation via Spatial Adjacency}
$E_t \gets \emptyset$\;
\For{each $v \in V_t$ and $u \in U_t$}{
    \If{$v.\text{zone} == u.\text{zone}$ \textbf{or} $A[v.\text{zone}, u.\text{zone}] == 1$}{
        $E_t \gets E_t \cup \{(v, u)\}$\;
    }
}
\BlankLine

\Return $G_t = (V_t \cup U_t, E_t)$\;
\end{algorithm2e}

\subsection{Monte Carlo construction of spread selection weights.}\label{sec:MonteCarlo}
Our theoretical guarantee depends on the geometry of the fractional solution used to guide local sparsification: solutions whose mass is spread across many light edges incur smaller collision penalties. In the experiments, we therefore use a simple Monte Carlo heuristic, described in \Cref{alg:offline_mc}, to construct spread fractional weights. We repeatedly sample a full realization of the stochastic instance, compute an offline maximum matching, and average the resulting matched-edge incidences. Because the matching algorithm is deterministic, we randomly shuffle the ordering of vertices before each run; this acts as random tie-breaking among multiple optimal matchings and helps spread mass across interchangeable or nearly interchangeable resources.

Formally, let $G$ denote a random realization of the full graph, and let $M(G)$ be an offline maximum matching selected using this randomized tie-breaking rule. Let $A_j(G)$ be the number of arrivals of type $j$ in $G$, and let $C_{ij}(G)$ be the number of type-$j$ arrivals matched to resource $v_i$ in $M(G)$. The expected edge incidence is
\[
y_{ij} := \mathbb E[C_{ij}(G)].
\]
Equivalently, since $\mathbb E[A_j(G)]=np_j$, the corresponding conditional LP variable is
\[
x_{ij} := \frac{y_{ij}}{np_j} = \Pr[\text{a type-}j\text{ arrival is matched to }v_i],
\]
where the probability is over the stochastic realization and the randomized tie-breaking. The vector $x$ is feasible for the expected-instance LP: every realized matching uses each resource at most once and matches each arrival at most once, so the resource and type constraints hold after taking expectations. Moreover,
\[
\sum_{i,j} np_j x_{ij} = \sum_{i,j} y_{ij} = \mathbb E[|M(G)|].
\]
Thus the Monte Carlo estimates approximate a feasible fractional solution of value $Z=\mathbb E[|M(G)|]$, and these estimates are used as the local
selection weights.

\begin{algorithm2e}[H]
\LinesNumbered
\DontPrintSemicolon
\caption{Monte Carlo Weight Estimation with Random Tie-Breaking}
\label{alg:offline_mc}
\KwIn{Rider distribution $D$, supply $V$, simulation count $M$, graph size $n$}
\KwOut{Conditional selection weights $x_{ij}$ for type-resource edges $(u_j,v_i)$}
\BlankLine

Initialize counters $C_{ij} \gets 0$ for all valid type-resource pairs $(u_j,v_i)$\;
Initialize counters $A_j \gets 0$ for all request types $t_j$ \tcp*{number of arrivals of type $j$ across simulations}
\For{$m = 1, \dots, M$}{
    Generate mock arrivals $U^{(m)}$ by sampling $n$ riders from $D$\;
    Construct full bipartite graph $G^{(m)} = (U^{(m)} \cup V, E^{(m)})$\;
    \BlankLine
    \tcp{Randomize vertex order to implement tie-breaking and promote spread}
    Randomly shuffle the order of vertices in $U^{(m)}$ and $V$\;
    $M^{(m)} \gets \mathrm{HopcroftKarp}(G^{(m)})$\;
    \BlankLine
    \For{each arriving rider $u$ of type $t_j$ in $U^{(m)}$}{
        $A_j \gets A_j+1$\;
        \If{rider is matched (there is $v_i\in \Gamma(t_j)$ such that $(u,v_i) \in M^{(m)}$)}{$C_{ij} \gets C_{ij}+1$\;}
    }
}
\BlankLine
\For{each request type $t_j$}{
    \For{each $v_i\in \Gamma(t_j)$}{
    $x_{ij} \gets (\sum_{i'} C_{i'j}) > 0 \ ? \ C_{ij} / A_j : 1/|\Gamma(t_j)|$
\tcp*{Implementation-only uniform fallback for unseen or seen and unmatched types}
    }
}
\Return $\{x_{ij}\}$\;
\end{algorithm2e}

\begin{algorithm2e}[H]
\LinesNumbered
\DontPrintSemicolon
\caption{\varopt~ Local Sparsifier Construction (Online)}
\label{alg:online_sparsifier}
\KwIn{Arriving rider $u$ of type $t_j$, Budget $k$, Pre-computed weights $\{x_{ij}\}$}
\KwOut{Selected edges $S_u \subseteq \{(u,v)\mid v\in \Gamma(t_j)\}$}
\BlankLine

Identify all valid candidate edges $E_u = \{(u, v_i) : v_i \in \Gamma(t_j)\}$\;
\For{each $e=(u, v_i) \in E_u$}{
    $\tilde{x}_e \gets x_{ij}$ 
}

\BlankLine
\tcp{\varopt$_k$ weighted sampling by $\tilde{x}_e$}
$S_u \gets \varopt_k((e,\tilde{x}_e)_{e\in E_u})$ \tcp*{Apply the stream-based algorithm \citep{varopt2011}}

\BlankLine
\tcp{The threshold $\tau$ such that $\sum_{e \in E_u} \min(1, \tau \cdot \tilde{x}_{e}) = \min(k, |E_u|)$ and exact inclusion probabilities $\pi_{e} = \min(1, \tau \cdot \tilde{x}_{e})$ are realized by this procedure.}

\Return $S_u$\;
\end{algorithm2e}

\end{document}